\documentclass[11pt]{article}
\usepackage{amsmath,amssymb,amsbsy,amsfonts,amsthm,latexsym,
               amsopn,amstext,amsxtra,euscript,amscd,color}
\def \F {{\mathbb F}}
\def \Q {{\mathbb Q}}
\def \Z {{\mathbb Z}}

\def \V {{\mathbb V}}

\def \T {{\rm Tr}}

\newtheorem{theorem}{Theorem}
\newtheorem{definition}[theorem]{Definition}
\newtheorem{lemma}[theorem]{Lemma}
\newtheorem{proposition}[theorem]{Proposition}
\newtheorem{remark}[theorem]{Remark}

\newtheorem{corollary}[theorem]{Corollary}

\def\cH{{\mathcal H}}

\def\cGB{\mathcal{GB}}

\def\00{{\bf 0}}
\def\11{{\bf 1}}
\def\+{\oplus}

\def \F {{\mathbb F}}
\def \Q {{\mathbb Q}}
\def \Z {{\mathbb Z}}

\def \V {{\mathbb V}}

\def \T {{\rm Tr}}

\begin{document}

\title{\huge\bf
\textrm{Partial Spread and Vectorial Generalized Bent Functions} }

\author{\Large  Thor Martinsen$^1$, Wilfried Meidl$^2$, \and \Large  Pantelimon St\u anic\u a$^1$
\vspace{0.4cm} \\
\small $^1$Department of Applied Mathematics, \\
\small Naval Postgraduate School, Monterey, CA 93943-5212, U.S.A.;\\
\small Email: {\tt \{tmartins,pstanica\}@nps.edu}\\
\small $^2$Johann Radon Institute for Computational and Applied Mathematics,\\
\small Austrian Academy of Sciences, Altenbergerstrasse 69, 4040-Linz, Austria;\\
\small Email: {\tt meidlwilfried@gmail.com}
}

\date{\today}
\maketitle
\thispagestyle{empty}

\begin{abstract}
In this paper we generalize the partial spread class and completely describe it for generalized Boolean functions from $\F_2^n$ to $\mathbb{Z}_{2^t}$.
Explicitly, we describe gbent functions from $\F_2^n$ to $\mathbb{Z}_{2^t}$, which can be seen as a gbent version of Dillon's $PS_{ap}$ class.
For the first time, we also introduce the concept of a vectorial gbent function from $\F_2^n$ to $\Z_q^m$, and determine the maximal value which $m$ can attain for the  case $q=2^t$.
Finally we point to a relation between vectorial gbent functions and relative difference sets.
\end{abstract}

\section{Introduction}

Let $\V_n$ be the $n$-dimensional vector space over the two-element field $\F_2$ and for an integer $q$ let $\Z_q$ be the ring of integers modulo $q$.
For a function $f$ from $\V_n$ to $\Z_q$ the generalized Walsh-Hadamard transform is the complex valued function
\[ \mathcal{H}^{(q)}_f(u) = \sum_{x\in \V_n}\zeta_q^{f(x)}(-1)^{\langle u,x\rangle},\quad  \zeta_q = e^{\frac{2\pi i}{q}}, \]
where $\langle,\rangle$ denotes a nondegenerate inner product on $\V_n$ (we shall use $\zeta$, respectively, $\cH_f$, instead of $\zeta_q$, respectively, $\cH_f^{(q)}$,  when $q$ is fixed).
We will follow our notations from~\cite{MMS1} and denote the set of all generalized Boolean functions by $\mathcal{GB}_n^q$ and when $q=2$, by $\mathcal{B}_n$.
A function $f\in \mathcal{GB}_n^q$ is called {\em generalized bent} ({\em gbent}) if $|\mathcal{H}^{(q)}_f(u)| = 2^{n/2}$ for all $u\in \V_n$.
Recall that if $q=2$, these functions are called {\em bent}.

If $f$ is gbent such that for every $u\in \V_n$, we have $\mathcal{H}^{(q)}_f(u) = 2^{n/2}\zeta_q^{j_u}$ for some $0\le j_u < q$,
then - following the notation for bent functions in odd characteristic (see \cite{agw,KSW85}) - we call $f$ a {\em regular} gbent function.
Similar as for bent functions we call $f^*$ the {\em dual} of $f$, if $2^{n/2}\zeta_q^{f^*(u)} = \mathcal{H}^{(q)}_f(u)$.
With the same argument as for the conventional bent functions we can see that the dual $f^*$
is also gbent and $(f^*)^* = f$. Hence regular gbent functions always appear in pairs.
First note that for $y\in\V_n$ we have
\begin{align*}
& \sum_{u\in\V_n}(-1)^{\langle u,y\rangle}\mathcal{H}^{(q)}_f(u) =
\sum_{u\in\V_n}(-1)^{\langle u,y\rangle}\sum_{x\in\V_n}\zeta_q^{f(x)}(-1)^{\langle b,x\rangle} = \\
& \sum_{x\in\V_n}\zeta_q^{f(x)}\sum_{u\in\V_n}(-1)^{\langle u,x+y\rangle} = 2^n\zeta_q^{f(y)}.
\end{align*}
With $\mathcal{H}^{(q)}_f(u) = 2^{n/2}\zeta_q^{f^*(u)}$, we then get
\[ 2^n\zeta_q^{f(y)} = 2^{n/2}\sum_{u\in\V_n}(-1)^{\langle u,y\rangle}\zeta_q^{f^*(u)}. \]
%
%
We finally remark that as shown in \cite{MMS1}, gbent functions from $\V_n$ to $\Z_{2^t}$, $t\ge 1$, which are the functions in which we are
most interested in this article, are always regular. Therefore the dual of a gbent function is always defined and it is a gbent function, as well.

Since the introduction of Boolean bent functions in~\cite{r}, bent functions and generalizations, like bent functions in odd characteristic,
negabent functions and the more general class of gbent functions (see e.g. \cite{HP,KUS1,KUS2,smgs}), attracted a lot of attention.
Many classes of bent functions have been proposed, the most famous being the Maiorana-McFarland class and Dillon's partial spread~($PS$)
class~\cite{d}. In this article we generalize the partial spread class to gbent functions. In Section~\ref{secPSap} we explicitly describe
gbent functions in $\mathcal{GB}_n^{2^t}$, which can be seen as a gbent version of Dillon's $PS_{ap}$ bent functions, which
form a subclass of the class of partial spread bent functions. In Section~\ref{secPSpm} we give a complete characterization of the
partial spread class for gbent functions in $\mathcal{GB}_n^{2^t}$. We suggest a concept of vectorial gbent functions from $\F_2^n$ to~$\Z_q^m$
in Section~\ref{vectorGB}, and determine the maximal value which $m$ can attain for~$q=2^t$.
We show that our bound for $m$ is attained giving an example of vectorial gbent functions arising from the class of partial spread gbent functions.
Finally we point to a relation between vectorial gbent functions and relative difference sets.

\section{$\mathcal{PS}_{ap}$ gbent functions}
\label{secPSap}

In \cite{smgs} the following construction of gbent functions has been introduced and referred to as the generalized Dillon class:
Let $n=2m$, and let $U_0, U_1,\ldots$, $U_{2^m}$ be a spread of $\V_n$, that is, $U_i$'s, $0\le i\le 2^m$, are $m$-dimensional subspaces of
$\V_n$ with pairwise trivial intersection. For integers $k_0,k_1,\ldots,k_{2^m},r$ of the set $\{0,1,\ldots,q-1\}$ such that
$\sum_{i=0}^{2^m}\zeta_q^{k_i} = \zeta_q^r$, we define $f:\V_n\rightarrow\Z_q$ as
\begin{equation}
\label{gPS}
f(x) = k_i\quad\mbox{if}\quad x\in U_i\quad\mbox{and}\quad x\ne 0,\quad\mbox{and}\quad f(0) = r.
\end{equation}
The gbentness of $f$ follows easily from the fact that for every nonzero $u\in \V_n$ we have $\langle u,x\rangle = 0$
for all $x\in U_t$, for exactly one $0\le t\le 2^m$. On the other spread elements $\langle u,x\rangle$ is balanced. If $u\ne 0$, then
\begin{eqnarray*}
\mathcal{H}^{(q)}_f(u) & = & \sum_{x\in \V_n}\zeta_q^{f(x)}(-1)^{\langle u,x\rangle} =
\sum_{i=0}^{2^m}\sum_{x\in U_i}\zeta_q^{k_i}(-1)^{\langle u,x\rangle} - \sum_{i=0}^{2^m}\zeta_q^{k_i} + \zeta_q^r \\
& = & \sum_{i=0}^{2^m}\zeta_q^{k_i}\sum_{x\in U_i}(-1)^{\langle u,x\rangle} = 2^{n/2}\zeta_q^{k_t},
\end{eqnarray*}
if $u\in U_t^\perp$ ($U^\perp$ is the orthogonal complement of $U$, with respect to ${\langle u,x\rangle}$). If $u=0$, then
\[ \mathcal{H}^{(q)}_f(0) = \sum_{i=0}^{2^m}\zeta_q^{k_i}2^{n/2} = 2^{n/2}\zeta_q^r. \]
We observe that $f$ defined in $(\ref{gPS})$ is a regular gbent function and that the dual $f^*$ of $f$ is defined via the orthogonal
spread (with respect to the inner product $\langle,\rangle$) as
\[ f^*(x) = k_i\quad\mbox{if}\quad x\in U^\perp_i\quad\mbox{and}\quad x\ne 0,\quad\mbox{and}\quad f^*(0) = r. \]

For $q=2$ the subclass of bent functions obtained with the construction in~\eqref{gPS} using the regular (Desarguesian) spread is called
Dillon's $PS_{ap}$ class. To be precise, we obtain a $PS^-$ bent function defined on the regular spread if $r=0$, and if $r=1$ we obtain
the complement of a $PS^-$ bent function, which in this case it is a $PS^+$ bent function. For the definition of $PS^-$ and $PS^+$ bent functions
we refer to \cite{d}.

In bivariate form, that is, as functions from $\F_{2^m}\times\F_{2^m}$ to $\F_2$, the $PS_{ap}$ class has an explicit
representation as $f(x,y) = G\left(\frac{x}{y}\right)$ for a balanced function $G:\F_{2^m}\rightarrow \F_2$ (we always assume the convention that~$1/0 = 0$).
In the following theorem we present an explicit representation of functions in a generalization of Dillon's $PS_{ap}$ class to
gbent functions with $q=2^k$. We use $\mathcal{H}$ for $\mathcal{H}^{(2^k)}$, and $\zeta = e^{\frac{2\pi i}{2^k}}$.
\begin{theorem}
\label{exPSk}
Let $G_j:\F_{2^m}\rightarrow\F_2$, $0\leq j\leq k-1$, be  
Boolean functions with $G_j(0)=0$ and $\displaystyle\sum_{t\in\F_{2^m}}\zeta^{\sum_{j=0}^{k-1} 2^j G_j(t)}=0$.
Then the function
$f:\F_{2^m}\times\F_{2^m}\rightarrow\Z_{2^k}$ given by
\[ f(x,y) = \sum_{j=0}^{k-1} 2^j G_j\left({x}/{y}\right)   \]
is a gbent function with the dual
\[ f^*(x,y) = \sum_{j=0}^{k-1} 2^j G_j\left({y}/{x}\right).   \]
\end{theorem}

\begin{proof}
Using the inner product $\langle (x_1,y_1),(x_2,y_2)\rangle = \T_m(x_1x_2+y_1y_2)$ on $\F_{2^m}\times\F_{2^m}$, for $u,v\in\F_{2^m}$,
with $s:=y/x$ we have
\begin{eqnarray*}
\mathcal{H}_f(u,v) & = & \sum_{s\in\F_{2^m}}\sum_{\substack{x\in\F_{2^m}\\  x\ne 0}}\zeta^{\sum_{j=0}^{k-1} 2^j G_j(s^{-1})}(-1)^{\T_m(ux+vsx)}\\
 & &\qquad\qquad\qquad + \sum_{y\in\F_{2^m}}(-1)^{\T_m(vy)} \\
& = & \sum_{s\in\F_{2^m}}\zeta^{\sum_{j=0}^{k-1} 2^j G_j(s^{-1})}\sum_{x\in\F_{2^m}}(-1)^{\T_m(ux+vsx)} \\
& & - \sum_{s\in\F_{2^m}} \zeta^{\sum_{j=0}^{k-1} 2^j G_j(s^{-1})} + \sum_{y\in\F_{2^m}}(-1)^{\T_m(vy)} := I - II + III.
\end{eqnarray*}
By the assumption on the balanced functions $G_j$, then $II = 0$.
If $v\ne 0$, then $III = 0$, and consequently
\[ \mathcal{H}_f(u,v) = 2^m \zeta^{\sum_{j=0}^{k-1} 2^j G_j(v/u)}. \]
If $v=0$, then $III = 2^m$. Consequently, with $\sum_{x\in\F_{2^m}}(-1)^{\T_m(ux)} = 0$ if $u\ne 0$, we get $\mathcal{H}_f(u,0) = III = 2^m$.
Finally,
\[ \mathcal{H}_f(0,0) = 2^m \sum_{s\in\F_{2^m}} \zeta^{\sum_{j=0}^{k-1} 2^j G_j(s^{-1})} + 2^m = 2^m II + 2^m = 2^m \]
by the assumption on the functions $G_j$. Therefore, in all cases, $|\mathcal{H}_f(u,v)|=2^m$, hence $f$ is gbent.
As $\mathcal{H}_f(u,v)$ is obtained explicitly for all $(u,v)$, we also can confirm the formula for the dual.
\end{proof}
%
With $G_0(x) = \T_m(ax)$ and $G_1(x) = \T_m(bx)$ for two distinct elements $a,b\in\F_{2^m}^*$ we obtain the following corollary.
\begin{corollary}
Let $a,b\in\F_{2^m}^*$, $a\ne b$, then the function $f:\F_{2^m}\times\F_{2^m}\rightarrow\Z_4$
\[ f(x,y) = \T_m\left(\frac{ax}{y}\right) + 2\T_m\left(\frac{bx}{y}\right) \]
with the convention that $1/0 = 0$, is gbent.
\end{corollary}
\begin{proof}
For a function $G:\F_{2^m}\rightarrow\F_2$ we put $G^i :=\{x\in\F_{2^m}\;|\;G(x)=i\}$, $i=0,1$.
With this notation, $G^0_0$ and $G^0_1$ are two distinct hyperplanes of $\F_{2^m}$, which intersect in an $(m-2)$-dimensional subspace.
Consequently the condition $|G^{0}_0\cap G^{1}_0| = 2^{m-2}$ is satisfied, which further implies that $|G_0^j\cap G_1^k| = 2^{m-2}$,
for all $j,k\in\{0,1\}$, and so, $\sum_{s\in\F_{2^m}}i^{G_0(s^{-1})+2G_1(s^{-1})}=0$, and the previous theorem applies.
\end{proof}
%



%

\section{$\mathcal{PS}^{+/-}$ gbent functions}
\label{secPSpm}

Being defined on a complete spread, for $q=2$ with the construction in $(\ref{gPS})$ we obtain $PS^-$ or complements of $PS^-$ bent functions.
To generate a partial spread bent function, solely $2^{m-1}$ subspaces of dimension $m$ with pairwise trivial intersection are needed (see~\cite{d}).
Since, in general, a partial spread is not contained in a complete spread, many more bent functions are in the partial spread class.
In this section we generalize the concept of partial spread bent functions to gbent functions $f\in\mathcal{GB}_n^q$, $q=2^t$, by completely
characterizing all gbent functions for which
\begin{itemize}
\item[-] $f$ is constant on the nonzero elements of every element of a partial spread
$\{U_1,U_2,\ldots,U_A\}$,
\item[-] $f(0) = \rho$ for some $0\le \rho\le 2^t-1$, and
\item[-] $f(x) = 0$ for $x\in\V_n\setminus\bigcup_{k=1}^AU_k$.
\end{itemize}
Here we always assume that $f$ is constant {\it nonzero} on $U^*$, $1\le k\le A$. Otherwise we may switch to an according subspread by deleting some
of the $U_k$ from the partial spread.
We remark that such a generalization to bent functions in odd characteristic has been given in \cite{k,ll}.

Since $q=2^t$ is fixed, in this section we again write $\mathcal{H}$ for $\mathcal{H}^{(q)}$, and put $\zeta = e^{\frac{2\pi i}{2^t}}$.
\begin{proposition}
\label{PSt}
Let $q=2^t$, $n=2m$ and let $U_1,\ldots, U_A$ be elements of a partial spread of $\V_n = \V_{2m}$. For integers $k_1,k_2,\ldots,k_A$
of the set $\{1,\ldots,q-1\}$ and $0\le \rho \le q-1$, such that
\begin{equation}
\label{PScon}
\sum_{i=1}^A\zeta^{k_i} = A-(2^m+1)+\zeta^\rho
\end{equation}
we define a function $f$ from $\V_n$ to $\Z_q$ by
\begin{align*}
&f(0) = \rho \text{ and } f(x) = k_i\quad\mbox{if}\quad x\in U_i\quad\mbox{and}\quad x\ne 0,\\
& f(x) = 0 \quad\mbox{if}\quad x\in \V_n\setminus\bigcup_{k=1}^AU_k.
\end{align*}
The function $f$ is gbent, and the dual $f^*$ of $f$ is obtained with the orthogonal spread as
\begin{align*}
&f^*(0) = \rho \text{ and } f(x) = k_i\quad\mbox{if}\quad x\in U_i^\perp\quad\mbox{and}\quad x\ne 0,\\
& f(x) = 0 \quad\mbox{if}\quad x\in \V_n\setminus\bigcup_{k=1}^AU_k^\perp.
\end{align*}
\end{proposition}
\begin{proof}
Let $R$ be the set $R = \V_n\setminus\bigcup_{k=1}^AU_k$, which has cardinality $|R| = 2^n - A(2^m-1)-1$. Then
\begin{align*}
\mathcal{H}_{f}(0) &= \sum_{x\in \V_n}\zeta^{f(x)} = \sum_{i=1}^A\zeta^{k_i}\sum_{\substack{x\in U_i\\ x\ne 0}}1 + \zeta^\rho+\sum_{x\in R}1 \\
& = (2^m-1)(A-(2^m+1)+\zeta^\rho) + \zeta^\rho + 2^n-A(2^m-1)-1 = 2^m\zeta^\rho.
\end{align*}
To evaluate $\mathcal{H}_{f}(u)$ for $u\ne 0$, we distinguish two cases. First we suppose that $u$ is not an element of $U_r^\perp$ for any $1\le r\le A$. Then
$\sum_{k=1}^A\sum_{x\in U_k^*}(-1)^{\langle u,x\rangle} = -A$. Therefore, with $\sum_{x\in\V_n}(-1)^{\langle u,x\rangle} = 0$ we have
$\sum_{x\in R}(-1)^{\langle u,x\rangle} = A-1$. As a consequence,
\begin{align*}
\mathcal{H}_{f}(u) &= \sum_{x\in \V_n}\zeta^{f(x)}(-1)^{\langle u,x\rangle} \\
&= \sum_{i=1}^A\zeta^{k_i}\sum_{x\in U_i}(-1)^{\langle u,x\rangle}
- \sum_{i=1}^A\zeta^{k_i} + \zeta^\rho + \sum_{x\in R}(-1)^{\langle u,x\rangle} \\
&= -\sum_{i=1}^A\zeta^{k_i}+\zeta^\rho+ A-1\\
& = -[A-(2^m+1)+\zeta^\rho]+\zeta^\rho + A -1 = 2^m.
\end{align*}
Now suppose that $u\in U_r^\perp$. In this case $\sum_{k=1}^A\sum_{x\in U_k^*}(-1)^{\langle u,x\rangle} = 2^m-A$, and hence
$\sum_{x\in R}(-1)^{\langle u,x\rangle} = A-2^m-1$. For $\mathcal{H}_{f}(u)$ we then get
\begin{align*}
\mathcal{H}_{f}(u) &= \sum_{i=1}^A\zeta^{k_i}\sum_{x\in U_i}(-1)^{\langle u,x\rangle} - \sum_{i=1}^A\zeta^{k_i} + \zeta^\rho - 2^m+A-1 \\
&= 2^m\zeta^{k_r} - A + 2^m+1-\zeta^\rho + \zeta^\rho - 2^m+A-1 = 2^m\zeta^{k_r}.
\end{align*}
Observing that $\mathcal{H}_f(0) = 2^m\zeta^\rho$, $\mathcal{H}_f(u) = 2^m$ if $u\not\in U_r^\perp$ for any $1\le r\le A$, and
$\mathcal{H}_f(u) = 2^m\zeta^{k_r}$ if $u\in U_r^\perp$, we confirm the statement for the dual $f^*$.
\end{proof}
%
The next corollary confirms that Proposition \ref{PSt} exactly yields the class of partial spread bent functions when $t=1$.
\begin{corollary}
\label{PS+-}
For $t=1$, the functions in Theorem~\textup{\ref{PSt}} are exactly the partial spread bent functions. In particular,
with $\rho = 0$ one obtains the class of the $PS^-$ Boolean bent functions, and with $\rho = 1$ one obtains the
class of the $PS^+$ Boolean bent functions.
\end{corollary}
\begin{proof}
If $t=1$, i.e. $q=2$, then $k_1 = k_2 = \cdots = k_A = 1$ and the condition~\eqref{PScon} is $\sum_{i=1}^A(-1)^1 = A-(2^m+1)+(-1)^\rho$,
or equivalently $2A = 2^m+1-(-1)^\rho$. Hence $A=2^{m-1}$ if $\rho = 0$, and $A=2^{m-1}+1$ if $\rho = 1$. In other words, if $\rho = 0$,
then the support of the Boolean function $f$ is the union of $2^{m-1}$ spread elements excluding the $0$, if $\rho = 1$, then the support
of the Boolean function $f$ is the union of $2^{m-1}+1$ spread elements (with the $0$). This exactly defines the class of the $PS^-$ Boolean
bent functions respectively the class of the $PS^+$ Boolean bent functions. Conversely, it is easily seen that any $PS^-$ (respectively, $PS^+$) Boolean bent
function is of the form of $f$ in Theorem~\ref{PSt} satisfying~\eqref{PScon} with $A = 2^{m-1}$ and $\rho = 0$
(respectively, $A = 2^{m-1}+1$ and $\rho = 1$).
\end{proof}
In the remainder of this section we show that Proposition~\ref{PSt} covers all gbent functions $f\in\mathcal{GB}_n^{2^t}$ which are constant on the nonzero
elements of every element of a partial spread, $f(0) = \rho$ for some $0\le \rho\le 2^t-1$, and $f(x) = 0$ for the remaining $x$. We may call this class the
class of the {\it partial spread gbent functions} in $\mathcal{GB}_n^{2^t}$. In Theorem~\ref{I-IV} below, we will represent this class of gbent functions in
a more descriptive way. We will use the following lemma.
\begin{lemma}
\label{L4Ny}
Let $q=2^t$, $t>1$, $\zeta=e^{2\pi i/q}$.
If $\rho_k\in\Q$, $0 \le k\le q-1$ and $\sum_{k=0}^{q-1}\rho_k\zeta^k = r$  is rational,
then $\rho_j = \rho_{2^{t-1}+j}$, for $1\le j \le 2^{t-1}-1$.
\end{lemma}
\begin{proof}
Since $\zeta^{2^{t-1}+k} = -\zeta^k$ for $0\le k \le 2^{t-1}-1$, we can write every element $z$
of the cyclotomic field $\Q(\zeta)$ as
\[ z = \sum_{k=0}^{2^{t-1}-1}\lambda_k\zeta^k,\,\lambda_k\in\Q, 0\le k\le 2^{t-1}-1. \]
As $[\Q(\zeta):\Q] = \varphi(q) = 2^{t-1}$, the set $\{1,\zeta,\ldots,\zeta^{2^{t-1}-1}\}$ is a basis of
$\Q(\zeta)$. Since
\[ 0 = \sum_{k=0}^{q-1}\rho_k\zeta^k - r = (\rho_0-\rho_{2^{t-1}}-r) + \sum_{k=1}^{{2^{t-1}-1}}(\rho_j - \rho_{2^{t-1}+j})\zeta^k. \]
the assertion of the lemma follows.
\end{proof}

We recall the next result shown in~\cite{MMS1}.
\begin{proposition}
\label{propreg}
All gbent functions in $\cGB_n^{2^t}$ are regular.
\end{proposition}
With Lemma \ref{L4Ny} we can also describe the distribution of the values of a gbent function in $\cGB_n^{2^t}$.
\begin{lemma}
For $q=2^t$, $n = 2m$, let $f\in\mathcal{GB}_n^q$ be a gbent function and for $j\in\Z_p$ denote $b_j:=|f^{-1}(j)|$.
Then there exists $0\le k \le 2^{t-1}-1$ such that
\[ b_{2^{t-1}+k} = b_k\pm 2^m \mbox{ and } b_{2^{t-1}+j} = b_j, \mbox{ for } 0\le k \le 2^{t-1}-1, j\ne k. \]
\end{lemma}
\begin{proof}
By Proposition \ref{propreg} the gbent function $f$ is regular. Hence for some $0\le k^\prime \le 2^t-1$,
\[ \mathcal{H}_f(0) = \sum_{x\in\V_n}\zeta^{f(x)} = \sum_{j=0}^{2^t-1}b_j\zeta^j = 2^m\zeta^{k^\prime}. \]
With $k^\prime = k$ or $k^\prime = 2^{t-1}+k$ for some $0\le k \le 2^{t-1}-1$ the claim follows then from Lemma~\ref{L4Ny}.
\end{proof}

The next theorem is the main result of this section. It completely describes the class of partial spread gbent functions in $\mathcal{GB}_n^{2^t}$.
\begin{theorem}
\label{I-IV}
Let $q=2^t$, $n=2m$ and let $U_1,\ldots, U_A$ be elements of a partial spread of $\V_n = \V_{2m}$. Let $f$ be constant on $U_k^*$, $1\le k\le A$, $f(0) = \rho$ for some
$0\le\rho\le 2^t-1$, and $f(x) = 0$ for $x\in\V_n\setminus\bigcup_{k=1}^AU_k$. We denote by $c_j$, $1\le j\le 2^t-1$, the number of spread elements whose nonzero elements are
mapped to $j$ and put $\sum_{j=1}^{2^{t-1}-1}c_j := \Delta$ and $c_{2^{t-1}} := \bar{c}$.
If $f$ is gbent, then $f$ satisfies one of the conditions $I, II, III$, or $IV$, depending upon the value of~$\rho$.
\begin{itemize}
\item[I.] $\rho = 0$, $c_{2^{t-1}+j}=c_j$, $1\le j\le 2^{t-1}-1$, and $A = 2^{m-1} + \Delta$, $\bar{c} = 2^{m-1} - \Delta = 2^m-A$.
\item[II.] $1\le\rho\le 2^{t-1}-1$, $c_{2^{t-1}+j}=c_j$, $1\le j\le 2^{t-1}-1$, $j\ne \rho$, $c_{2^{t-1}+\rho} = c_\rho-1$, and
$A = 2^{m-1} + \Delta$, $\bar{c} = 2^{m-1} + 1 - \Delta = 2^m+1-A$.
\item[III.] $\rho = 2^{t-1}$, $c_{2^{t-1}+j}=c_j$, $1\le j\le 2^{t-1}-1$, and $A = 2^{m-1} + 1 + \Delta$, $\bar{c} = 2^{m-1} + 1 - \Delta = 2^m+2-A$.
\item[IV.] $2^{t-1}+1\le\rho\le 2^t-1$, $c_{2^{t-1}+j}=c_j$, $1\le j\le 2^{t-1}-1$, $j\ne \rho$, $c_\rho = c_{\rho-2^{t-1}}+1$, and
$A = 2^{m-1} + 1 + \Delta$, $\bar{c} = 2^{m-1} - \Delta = 2^m+1-A$.
\end{itemize}
Conversely, every function $f:\V_n\rightarrow\Z_{2^t}$ described in $I, II, III$, or $IV$ is a partial spread gbent function.
\end{theorem}
\begin{proof}
By a straightforward computation, one can easily check that the conditions in $I, II, III$, and $IV$ imply~\eqref{PScon}.
Hence by Proposition \ref{PSt} all such functions are partial spread gbent functions.

It remains to show that every partial spread gbent function satisfies one of the conditions $I, II, III$, or $IV$.
First observe that for every partial spread gbent function $f\in\mathcal{GB}_n^{2^t}$, for some $0\le s\le 2^t-1$ we have
\[ 2^m\zeta^s = \mathcal{H}_f(0) = 2^n-A(2^m-1)-1 + \zeta^\rho + (2^m-1)\sum_{j=1}^{2^t-1}c_j\zeta^j, \]
or equivalently
\begin{equation}
\label{reint}
(2^m-1)\sum_{j=1}^{2^t-1}c_j\zeta^j + \zeta^\rho - 2^m\zeta^s = (2^m-1)[A-(2^m+1)].
\end{equation}
{\it Case $s=0$:} By equation $(\ref{reint})$ in this case $\zeta^\rho + (2^m-1)\sum_{j=1}^{2^t-1}c_j\zeta^j$ is an integer, thus by
Lemma \ref{L4Ny} the coefficients of $\zeta^j$ and $\zeta^{j+2^{t-1}}$ must be equal for $1\le j\le 2^{t-1}-1$.
Since in $(2^m-1)\sum_{j=1}^{2^t-1}c_j\zeta^j$ all coefficients are $0$ modulo $(2^m-1)$, we must have $\rho = 0$ or $\rho = 2^{t-1}$,
i.e. $\zeta^\rho = \pm 1$, and $c_{j+2^{t-1}} = c_j$, $1\le j\le 2^{t-1}-1$. Equation $(\ref{reint})$ then yields
\begin{equation}
\label{1-version}
-(2^m-1)\bar{c} = (2^m-1)[A-(2^m+1)] + 2^m - \zeta^\rho.
\end{equation}
As the left side is divisible by $(2^m-1)$, on the right side of the equation we require $\zeta^\rho = 1$, consequently $\rho = 0$.
Dividing by $2^m-1$ we then get $\bar{c} = 2^m-A$. With $A = \bar{c} + 2\sum_{j=1}^{2^{t-1}-1}c_j = \bar{c}+2\Delta$, we obtain
$\bar{c}=2^{m-1}-\Delta$ and $A=2^{m-1}+\Delta$, and we observe that $f$ satisfies $I$. \\[.3em]
{\it Case $s=2^{t-1}$:} In this case for equation $(\ref{1-version})$ we obtain
\[ -(2^m-1)\bar{c} = (2^m-1)[A-(2^m+1)] - 2^m - \zeta^\rho \]
and hence we require $\zeta^\rho = -1$, consequently $\rho = 2^{t-1} = s$.
Dividing by $2^m-1$ then yields $\bar{c} = 2^m+2-A$, which implies $\bar{c}=2^{m-1}+1-\Delta$ and $A=2^{m-1}+1+\Delta$. Therefore
$f$ satisfies $III$. \\[.3em]
{\it Case $1\le s\le 2^{t-1}$:} As the right side of equation $(\ref{reint})$,
\[ (2^m-1)\sum_{j=1}^{2^t-1}c_j\zeta^j + \zeta^\rho - 2^m\zeta^s, \]
is an integer, the coefficients of $\zeta^j$ and $\zeta^{j+2^{t-1}}$,
$1\le j\le 2^{t-1}-1$, must be the same by Lemma \ref{L4Ny}. Since in $(2^m-1)\sum_{j=1}^{2^t-1}c_j\zeta^j$
all coefficients, in particular those of $\zeta^s$ and $\zeta^{s+2^{t-1}}$, are $0$ modulo $(2^m-1)$, we require that $\rho = s$ or
$\rho = s+2^{t-1}$. Observing that $\zeta^s-2^m\zeta^s = (1-2^m)\zeta^s$, but $\zeta^{s+2^{t-1}}-2^m\zeta^s = -(1+2^m)\zeta^s = B\zeta^s$
with $B$ not divisible by $2^m-1$, we conclude that $\rho = s$.
With $c_j = c_{j+2^{t-1}}$, $1\le j\le 2^{t-1}-1$, $j\ne \rho$, equation $(\ref{reint})$ then yields
\[ -(2^m-1)\bar{c} + [(2^m-1)c_\rho - (2^m-1)c_{\rho+2^{t-1}} - (2^m-1)]\zeta^\rho = (2^m-1)[A-(2^m+1)]. \]
Consequently, $(2^m-1)(c_\rho - c_{\rho+2^{t-1}} - 1) = 0$, i.e. $c_{\rho+2^{t-1}} = c_\rho - 1$, and
$\bar{c} = 2^m+1-A$. Now $A = \bar{c} + \sum_{j=1}^{2^{t-1}-1}c_j + \sum_{j=2^{t-1}+1}^{2^t-1}c_j = \bar{c}+2\Delta-1$,
from which we get $A = 2^{m-1}+\Delta$ and $\bar{c} = 2^{m-1}+1-\Delta$. Hence $f$ satisfies $II$. \\[.3em]
Similarly one shows that if the parameter $s$ in equation $(\ref{reint})$ satisfies $2^{t-1}+1\le s\le 2^t-1$, then $f$ satisfies~$IV$.
\end{proof}
%
We remark that one can also easily show that every function $f$ which is constant on every $U_k^*$, for which
$f(x) = 0$ if $x\in\V_n\setminus\bigcup_{k=1}^AU_k$ and for which $(\ref{PScon})$ holds, also satisfies $I, II, III$, or $IV$,
depending upon the value of~$f(0)$. Consequently also Proposition~\ref{PSt} describes the whole set of partial spread
gbent functions from $\V_n$ to $\Z_{2^t}$.

%
{\it Example for $q = 4$.}
To construct a partial spread gbent function $f$ from $\V_n$ to $\Z_4$ from a partial spread $U_1,\ldots, U_A$ of $\V_n = \V_{2m}$,
we again denote the number of spread elements whose nonzero elements are mapped to $1,2$ and $3$ by $c_1,c_2$ and $c_3$, respectively.
Then we can choose
\begin{itemize}
\item[$I.$] $f(0) = 0$, $c_1 = A-2^{m-1}$, $c_2 = 2^{m-1}-c_1$ and $c_3=c_1$,
\item[$II.$] $f(0) = 1$, $c_1 = A-2^{m-1}$, $c_2 = 2^{m-1}-c_1+1$ and $c_3=c_1-1$,
\item[$III.$] $f(0) = 2$, $c_1 = A-2^{m-1}-1$, $c_2 = 2^{m-1}-c_1+1$ and $c_3=c_1$,
\item[$IV.$] $f(0) = 3$, $c_1 = A-2^{m-1}-1$, $c_2 = 2^{m-1}-c_1$ and $c_3=c_1+1$.
\end{itemize}
\begin{remark}
If $q=2$, in which case Theorem~\textup{\ref{I-IV}} describes Dillon's partial spread class, then $f(0)$ uniquely determines $A$
(and $c_1$). \\
For $q=4$, the number $A$ of spread elements and $f(0)$ uniquely determine
$c_1,c_2$ and $c_3$. Note that we require $A \ge 2^{m-1}$ in case $I$, and $A \ge 2^{m-1}+1$ in the cases $II, III, IV$.
\end{remark}

\section{Vectorial gbent functions}
\label{vectorGB}

Recall that a function $F$ from $\F_2^n$ to $\F_2^m$ given as
\[ F(x_1,x_2,\ldots,x_n) = \left(
\begin{array}{c}
f_1(x_1,x_2,\ldots,x_n) \\
f_2(x_1,x_2,\ldots,x_n) \\
\vdots \\
f_m(x_1,x_2,\ldots,x_n)
\end{array}
\right) \]
is called vectorial bent, if every nontrivial linear combination $\lambda_1f_1 + \lambda_2f_2 + \cdots + \lambda_mf_m$
is bent. In other words, $\{f_1,f_2,\ldots,f_m\}$ is a basis of an $m$-dimensional vector space of bent functions over $\F_2$.
Classical examples of vectorial bent functions arise from the Maiorana-McFarland class and the partial spread class,
see for instance \cite{cmp4,n,zp}.
As already shown by Nyberg in \cite[Theorem 3.2]{n}, for a vectorial Boolean bent function $m$ can be at most $n/2$.
We remark that this is different for vectorial bent functions from $\F_p^n$ to $\F_p^m$ for an odd prime $p$,
where we have $m\le n$ (see again \cite{n}). The vectorial bent functions (in odd characteristic) with $m=n$ are
the widely-noted planar functions.

In the following definition we suggest a concept for a vectorial gbent function.
To the best of our knowledge, this is the first treatment of gbentness for vectorial functions.
%
%
\begin{definition}
\label{DefVecgb}
For two integers $m,n$, a function from $\F_2^n$ to $\Z_q^m$ given as
\[ F(x_1,x_2,\ldots,x_n) = \left(
\begin{array}{c}
f_1(x_1,x_2,\ldots,x_n) \\
f_2(x_1,x_2,\ldots,x_n) \\
\vdots \\
f_m(x_1,x_2,\ldots,x_n)
\end{array}
\right) \]
is called a vectorial gbent function if $\{f_1,f_2,\ldots,f_m\}$ is a basis of a $\Z_q$-module of gbent functions isomorphic to $\Z_q^m$.
The functions $\lambda_1f_1 + \lambda_2f_2 + \cdots + \lambda_mf_m$, $(\lambda_1,\lambda_2,\ldots,\lambda_m)\ne 0\in\Z_q^m$, are called
the component functions of $F$.
\end{definition}

In this section we are once again interested in the case $q = 2^t$. As before we write $\mathcal{H}_f$ for
$\mathcal{H}_f^{(q)}$ and we put $\zeta = e^{\frac{2\pi i}{2^t}}$.
In the next theorem we determine the maximal value which $m$ can attain for a vectorial gbent function from
$\F_2^n$ to $\Z_{2^t}^m$. This generalizes Theorem 3.2 in \cite{n} to vectorial gbent functions.
%
%
%
%
\begin{theorem}
\label{nyberg}
For $q=2^t$ and an even integer $n > 2$, let $F$ be a vectorial gbent function from $\F_2^n$ to $\Z_q^m$. Then $m\le n/(2t)$.
\end{theorem}
\begin{proof}
For  an $m$-tuple $c = (c_1,\ldots,c_m)\in\Z_q^m$ we denote the component function $c\cdot F = c_1f_1+\cdots+c_mf_m$ by
$F_c$. Then
\[ \mathcal{H}_{F_c}(0) = \sum_{x\in\F_2^n}\zeta^{c\cdot F(x)} = \sum_{y\in\Z_q^m}a_y\zeta^{c\cdot y}, \]
where $a_y = |\{F(x) = y\;|\;x\in \F_2^n\}|$ for all $y\in\Z_q^m$. Putting $S=2^{-n/2}\sum_{c\ne 0}\mathcal{H}_{F_c}(0)$ we have
\[ 2^{n/2}S = \sum_{y\in\Z_q^m}a_y\sum_{c\ne 0}\zeta^{c\cdot y} = \sum_{\substack{y\in\Z_q^m\\ y\ne 0}}a_y(-1) + a_0(q^m-1), \]
where in the last step we use that $\sum_{c\ne 0}\zeta^{c\cdot y} = -1$ for all $y\ne 0$.

With $\sum_{y\in\Z_q^m}a_y = 2^n$ we get
\[ 2^{n/2}S = -\sum_{y\in\Z_q^m}a_y + 2^{tm}a_0 = -2^n + 2^{tm}a_0, \]
hence
\[ S = -2^{n/2}+2^{tm-n/2}a_0. \]
In the next step we show that $S$ is an odd integer. First note that $S$ is rational since $n$ is even.
Put
\[ \rho_k = |\{c\in\Z_q^m, c\ne 0\;:\;\mathcal{H}_{F_c}(0) = 2^{n/2}\zeta^k\}|. \]
Recall that by Proposition~\ref{propreg} all component functions $F_c$ are regular. Hence we have $\sum_{k=0}^{q-1}\rho_k = q^m-1$ and $S=\sum_{k=0}^{q-1}\rho_k\zeta^k$.
Because $S$ is rational, by Lemma~\ref{L4Ny} we have
\[ q^m-1 = \rho_0 + \rho_{2^{t-1}} + 2\sum_{k=1}^{2^{t-1}-1}\rho_k \]
and $S = \rho_0 - \rho_{2^{t-1}}$. Combining, we obtain that
\[ S = 2\rho_0 + 2\sum_{k=1}^{2^{t-1}-1}\rho_k - q^m+1 \]
is an odd integer. Finally with
\[ a_0 = 2^{n/2-tm}(S+2^{n/2}) \]
for some odd integer $S$, we see that $n/2 \ge tm$.
\end{proof}
In the usual representation of the classical examples of conventional vectorial bent functions, like Maiorana-McFarland and partial spread vectorial
bent functions from $\F_{2^m}\times\F_{2^m}$ to $\F_{2^m}$, where $m$ is at most $n/2$, one takes advantage from the fact that the vectorial bent functions
form vector spaces isomorphic to $\F_{2^m}$ as a vector space over $\F_2$. This is different for vectorial gbent functions, hence we cannot use the structure
of the finite field in an analogous way to obtain examples. In view of Theorem~\ref{nyberg} we are mostly interested in vectorial gbent functions from
$\F_2^n$ to $\Z_{2^t}^m$ with~$m=n/(2t)$. We give a construction with generalized Dillon type gbent functions, which guarantees the existence of such vectorial
gbent functions.
\begin{theorem}
Let $n,m,t$ be integers such that $m=n/(2t)$, and let $U_s$, $0\le s\le 2^{n/2}$, be the $2^{n/2}+1$ elements of a spread of $\V_n$.
Consider a bijection $\phi:\{1,2,\ldots,2^{n/2}\}\rightarrow\Z_{2^t}^m$
\[ \phi(s) = (\phi_1(s),\phi_2(s),\ldots,\phi_m(s))^T, \]
and define $f_j:\V_n\rightarrow\Z_{2^t}$, $1\le j\le m$, by
\begin{equation*}
\begin{split}
&f_j(x) = \phi_j(s)  \mbox{ if }  x\in U_s, 1\le s\le 2^{n/2},  \mbox{ and }  x\ne 0,\\
 & \mbox{and }   f_j(x) = 0 \mbox{ if } x\in U_0.
\end{split}
\end{equation*}
The function $F$ from $\V_n$ to $\Z_{2^t}^m$ given by
\[ F(x) = \left(
\begin{array}{c}
f_1(x) \\
f_2(x) \\
\vdots \\
f_m(x)
\end{array}
\right) \]
is a vectorial gbent function.
\end{theorem}
\begin{proof}
We have to show that every component function $F_c(x)=c_1f_1(x)+c_2f_2(x)+\cdots+c_mf_m(x)$, $c=(c_1,c_2,\ldots,c_m)\ne 0\in\Z_{2^t}^m$, is a gbent function.
First observe that if $x\in U_s$, $1\le s\le 2^{n/2}$, and $x\ne 0$, then
\[ F_c(x) = \sum_{j=1}^mc_j\phi_j(s) = c\cdot\phi(s) = \left(\begin{array}{c}c_1 \\ c_2 \\ \vdots \\ c_m \end{array}\right)
\cdot
\left(\begin{array}{c}\phi_1(s) \\ \phi_2(s) \\ \vdots \\ \phi_m(s) \end{array}\right). \]
For $x\in U_0$ we have $F_c(x) = 0$. For $u\in \V_n$, we then obtain
\begin{eqnarray*}
\mathcal{H}_{F_c}(u) & = & \sum_{s=1}^{2^{n/2}}\sum_{\substack{x\in U_s\\ x\ne 0}}\zeta^{c\cdot\phi(s)}(-1)^{\langle u,x\rangle} + \sum_{x\in U_0}(-1)^{\langle u,x\rangle} \\
& = & \sum_{s=1}^{2^{n/2}}\zeta^{c\cdot\phi(s)}\sum_{x\in U_s}(-1)^{\langle u,x\rangle} - \sum_{s=1}^{2^{n/2}}\zeta^{c\cdot\phi(s)} + \sum_{x\in U_0}(-1)^{\langle u,x\rangle}.
\end{eqnarray*}
Since $\phi$ is a bijection, we have $\sum_{s=1}^{2^{n/2}}\zeta^{c\cdot\phi(s)} = 0$ for all nonzero $c\in\Z_{2^t}^m$.
Consequently for $u\ne 0$ we get
\[ \mathcal{H}_{F_c}(u) = \left\{
\begin{array}{l@{\quad:\quad}l}
2^{n/2} & u\in U_0^\perp, \\
2^{n/2}\zeta^{c\cdot\phi(\tilde{s})} & u\in U_{\tilde{s}}^\perp\:\mbox{for some}\;1\le \tilde s\le 2^{n/2}.
\end{array}
\right. \]
Again using that $\sum_{s=1}^{2^{n/2}}\zeta^{c\cdot\phi(s)} = 0$, we obtain $\mathcal{H}_{F_c}(0) = 2^{n/2}$, and the theorem is shown.
\end{proof}

Besides from applications in cryptography, one motivation for considering (vectorial) bent functions is their relation to objects in combinatorics.
For instance, a vectorial bent function from $\V_n$ to $\F_2^m$ gives rise to a relative difference set of $\V_n\times\F_2^m$.
We conclude this section pointing out a relation between relative difference sets and vectorial gbent functions as introduced in our in Definition
\ref{DefVecgb}. First we recall the definition of a relative difference set. Let $G$ be a group of order $\mu\nu$, let $N$ be a subgroup of $G$ of
order $\nu$ and let $R$ be a subset of $G$ of cardinality $k$. Then $R$ is called a $(\mu,\nu,k,\lambda)$-relative difference set of $G$ relative to
$N$, if every element $g\in G\setminus N$ can be represented in exactly $\lambda$ ways as difference $r_1-r_2$, $r_1,r_2\in R$, and no nonzero
element of $N$ has such a representation.

Relative difference sets can be described with characters as follows (see for instance Section 2.4. in \cite{tpf}).
\begin{proposition}
\label{RDS}
Let $G$ be an (abelian) group of order $\mu\nu$ and let $N$ be a subgroup of $G$ of order $\nu$.
A subset $R$ of $G$ (with $k$ elements) is an $(\mu,\nu,k,\lambda)$-relative difference set of $G$ relative to $N$ if and only if
for every character $\chi$ of $G$
\begin{equation*}
|\chi(R)|^2 = \left\{\begin{array}{ll}
                  k^2\quad & \mbox{if}\;\chi = \chi_0, \\
                  k-\lambda \nu\quad & \mbox{if}\;\chi \ne \chi_0,\:\mbox{but}\;\chi(g) = 1\;\mbox{for all}\;g\in N, \\
                  k\quad & \mbox{otherwise.}
                  \end{array} \right.
\end{equation*}
\end{proposition}
With this characterization of relative difference sets, the relation with our vectorial gbent functions becomes transparent.
Since it is the most interesting case, we consider a vectorial gbent function from $V_n$ to $\F_{2^t}^m$ with maximal
possible $m=n/(2t)$.
\begin{theorem}
Let $q=2^t$, and let $F$ be a vectorial gbent function from $\V_n$ to $\Z_q^m$, $m=n/(2t)$. Then the set
\[ R = \{(x,F(x))\;:\; x\in \V_n\} \]
is a $(2^n,2^{n/2},2^n,2^{n/2})$-relative difference set in $\V_n\times\Z_q^m$ relative to $\{0\}\times\Z_q^m$.
\end{theorem}
\begin{proof}
The theorem essentially follows from Proposition \ref{RDS} with definition of vectorial gbent functions
(it is the same argument as for the conventional vectorial bent functions, where $t=1$):
Note that the group of characters of $\V_n\times\Z_q^m$ consists of the elements
$\chi_{u,c}:(x,z)\rightarrow (-1)^{\langle u,x\rangle}\zeta^{c\cdot z}$, $u\in \V_n, c\in\Z_q^m$. Therefore
\[ \chi_{u,c}(R) = \sum_{x\in\V_n}(-1)^{\langle u,x\rangle}i^{c\cdot F(x)} = \mathcal{H}_{F_c}(u). \]
(We include now also $c=0$.) By the definition of a vectorial gbent function we then have
\begin{equation*}
|\chi_{u,c}(R)|^2 = |\mathcal{H}_{F_c}(u)|^2 = \left\{\begin{array}{ll}
                  2^{2n}\quad & \mbox{for}\; \chi_{0,0}, \\
                  0 \quad & \mbox{for}\; \chi_{u,0}, u\ne 0 \\
                  2^n\quad & \mbox{otherwise.}
                  \end{array} \right.
\end{equation*}
Hence by Proposition \ref{RDS}, the set $R$ is a $(2^n,2^{n/2},2^n,2^{n/2})$-relative difference, relative to $\{0\}\times\Z_q^m$.
\end{proof}
\begin{remark}
Differently to the case of bent functions, for a gbent function $f\in\mathcal{GB}_n^{2^t}$ the set $\{(x,f(x))\;:\;x\in\V_n\}$
is in general not a relative difference set (of $V_n\times \Z_{2^t}$ relative to $\{0\}\times \Z_{2^t}$).
For instance we may have $\tilde{c}f = 0$ for a nonzero $\tilde{c}\in\Z_{2^t}$ - $f$ is then not vectorial (defined as in Definition \ref{DefVecgb}
with $m=1$). The character sum that corresponds to $\chi_{u,\tilde{c}}$ does then not attain the required value.
An example is the gbent function in \cite[Theorem 8]{smgs}.
\end{remark}

\noindent
{\bf Acknowledgements.} Work by P.S. started during a very enjoyable visit at RICAM (Johann Radon Institute for Computational and Applied Mathematics),
Austrian Academy of Sciences, in Linz, Austria. Both the second and third named author thank the institution for the excellent working conditions. \\[.3em]
The second author is supported by the Austrian Science Fund (FWF) Project no. M 1767-N26.

\end{document}